\documentclass[reqno, centertags, 11pt]{amsart}
\usepackage{amsmath,amsthm,amsfonts,amssymb,graphicx}
\usepackage{srcltx}
\usepackage[all]{xy}
\usepackage{bussproofs}
\usepackage{pstricks} 
\usepackage{ifthen}
\usepackage{ifdraft}
\usepackage{changebar} 
\ifdraft{\usepackage[notref,notcite]{showkeys}}{}

\newcommand{\nl}{\par\noindent}   

\bibliography{plain}

\newtheorem{theo}{Theorem}[section]

\newtheorem{lem}[theo]{Lemma}

\begin{document}

\title[Continuous functions as quantum operations]
{\textbf{Continuous functions as quantum operations: a probabilistic approximation}}
\author[H.~Freytes]{Hector Freytes}
\address[H.~Freytes]{Instituto Argentino de Matem\'atica (CONICET), Saavedra 15,  AR-1083,
Buenos Aires, Argentina.}


\author[A.~Ledda]{Antonio Ledda}

\address[H.~Freytes, A.~Ledda, G.~Sergioli]{    Universit\`a di Cagliari, via Is Mirrionis 1, I-09123 Cagliari, Italy.}
\author[G.~Sergioli]{Giuseppe Sergioli}


\thanks{Corresponding author: A.~Ledda, \tt{antonio.ledda78@gmail.com}}
\date{\today}
\keywords{Quantum operations, PMV-algebras, quantum computation.}


\begin{abstract}
In this note we propose a version of the classical Stone-Weierstrass theorem in the context of quantum operations, by introducing a particular class of quantum operations, dubbed \emph{polynomial quantum operations}. This result permits to interpret from a probabilistic point of view, and up to a certain approximation, any continuous function from the real cube $[0,1]^n$ to the real interval $[0,1]$ as a quantum operation. 
\end{abstract}
\maketitle
\section*{Introduction}
\noindent Since the classical work of Birkhoff and von Neumann \cite{Birkhoff et al. 36}, logical and algebraic perspectives of several aspects of quantum theory have been proposed. Leading examples are orthomodular lattices \cite{Kal}, and effect algebras, appeared independently under several names (e.g. D-posets \cite{KC}) as a generalization of orthomodular posets \cite{GG, Gu}. Moreover, effect algebras play a fundamental role in various studies on fuzzy probability theory \cite{BHS, GUD-fuz} also.\\
In recent times, quantum computation itself stirred increasing attention, and an array of related algebraic structures arose \cite{CDGL1, DGG, GUD1,GLSP}. Those structures stem from an abstract description of circuits obtained by combinations of quantum gates \cite{Dalla Chiara et al. 2005}. Let us mention as examples {\it quantum MV-algebras} \cite{Giunt}, {\it quasi MV-algebras}, {\it $\sqrt{'}$quasi MV-algebras} and \emph{product MV-algebras} \cite{LKPG, GLP, EGM, MR}.
\nl Even if those structures are plainly related to quantum computing, some of the functions in their types are algebraic abstractions of irreversible transformations, e.g. the truncated disjunction ``$\oplus$'' \cite{Dalla Chiara et al. 2005}: the reference frame shifts from quantum gates to quantum operations \cite{AKN}. In the present paper we show that all those algebraic structures can be fully settled into the general model of quantum computing, based on quantum operations acting on density operators. 
To this aim, we propose a simple construction on density matrices (dubbed \emph{polynomial operations}) that permits to resort, in terms of probability distributions, to a Stone-Weierstrass type theorem. That result implies that any continuous function can be regarded, from a probabilistic point of view, and up to a certain approximation, as a quantum operation, Theorems \ref{KRO3} and \ref{KRO4}. 
The irreversible quantum operational approach has plenty of advantages in the implementation of quantum computational devices: as Aharonov, Kitaev and Nisan have discussed \cite{AKN}, there are several problems very difficult or impossible to deal in the usual unitary model of quantum computation.  On the other hand, these problems (such as measurements in the middle of computation or noise and decoherence ) disappear in the no-unitary (i.e. no-reversible) model. In fact, although quantum computations may allow measurements in the middle of the computation, however, the state of the computation after a measurement is a mixed state.  Moreover, in order to implement quantum computers devices, noise and, in particular, decoherence are serious obstacles. The main problem in this interface between quantum physics and quantum computation models is the fact that quantum noise and decoherence are non-unitary operations that cause a pure state to become a mixed state. 

\nl This paper is organized as follows: in Section \ref{bn} we provide all the basic notions, in Section \ref{sw} we show an overview of a probabilistic version of the Stone-Weierstrass theorem in the framework of quantum operations (further detalis are showed in \cite{FSA}), in Section \ref{rcn} some applications of our main result to the case of product MV-algebra \cite{MR} are given, and lastly in Section \ref{final}  some possible future investigation issues are illustrated and a few conclusive remarks drawn. 
\section{Basic notions}\label{bn}
\nl A quantum system in a pure state is described by a unit vector in a
Hilbert space. In the Dirac notation a pure state is denoted by $\vert \psi \rangle$ ($\langle \psi\vert$). A {\it quantum bit} or {\it qubit}, the fundamental concept of quantum computation, is a pure state in the Hilbert space $\mathbb{C}^2$. The standard orthonormal basis $\{ \vert 0\rangle , \vert 1 \rangle \}$ of $\mathbb{C}^2$ is called the {\it logical
basis}. Thus a qubit $\vert \psi \rangle$ may be written as a linear superposition of the basis vectors with complex
coefficients $\vert \psi \rangle = c_0\vert 0 \rangle + c_1 \vert 1 \rangle$ with  $\vert c_0 \vert^2 + \vert c_1 \vert^2 = 1 $. 
\nl Quantum mechanics reads out the information content of a pure state via the Born rule, according to which the probability value assigned to a qubit is defined as follows: 
$$\mathtt{p}(\vert\psi \rangle) = \vert c_1 \vert ^2.$$
\nl The states of interest  for  quantum computation lie in the
tensor product $\otimes^n \mathbb{C}^2 = \mathbb{C}^2\otimes \mathbb{C}^2 \otimes \cdots \otimes
\mathbb{C}^2$, where $\otimes^n \mathbb{C}^2 =\mathbb{C}^2$ if $n=1$. The space $\otimes^n \mathbb{C}^2$ is a $2^n$-dimensional complex Hilbert space. The $2^n$-{\it computational basis} consists of the $2^n$ orthogonal states $\vert \iota \rangle$ $(0 \leq \iota \leq 2^n)$, where $\iota$ is in binary representation and it can be
seen as the tensor product of the states $ \vert \iota_1 \rangle \otimes \vert \iota_2 \rangle \otimes \cdots \otimes
\vert \iota_n \rangle$ where $\iota_j \in \{0,1\}$. 
A pure state $\vert \psi \rangle \in \otimes^n \mathbb{C}^2$ is generally a superposition of
the basis vectors: $\vert \psi \rangle = \sum_{\iota = 1}^{2^n}
c_{\iota}\vert \iota \rangle$ with $\sum_{\iota = 1}^{2^n} \vert
c_{\iota} \vert^2 = 1$.
\nl In general, a quantum system is not in a pure state. This might be because the system is coupled with an environment, it is subject to a measurement process etc. In those cases, the state-evolution is no longer reversible and the system is said to be in a {\it mixed state}. A convenient mathematical description of a mixed state is given by the notion of \emph{\ density operator}, i.e. an Hermitian positive operator $\rho$ on a $2^n$-dimensional complex Hilbert space with trace $\mathtt{tr}(\rho )=1$. 


\nl A pure state $\vert \psi
\rangle$ can be represented as a limit case of mixed state in the form $\rho = \vert \psi
\rangle \langle \psi \vert$. In particular, each vector of the
logical basis of $\mathbb{C}^2$ can be associated to a density
operator $P_0 := \vert 0 \rangle \langle 0 \vert$ or $P_1 := \vert 1
\rangle \langle 1 \vert$ that represents the {\it falsity-property} and the
{\it truth-property}, respectively. One can represent an arbitrary density matrix $\rho$ in
terms of a tensor products of the Pauli matrices: 
$$ \sigma_0 = \left(\begin{array}{cc}
1 & 0 \\
0 & 1
\end{array}\right)
\hspace{0.5cm} \sigma_x = \left(\begin{array}{cc}
0 & 1 \\
1 & 0
\end{array}\right)
\hspace{0.5cm} \sigma_y = \left(\begin{array}{cc}
0 & -i \\
i & 0
\end{array}\right)
\hspace{0.5cm} \sigma_z = \left(\begin{array}{cc}
1 & 0 \\
0 & -1
\end{array}\right)
$$
\nl in the following way:  

$$\rho = \frac{1}{2^n} \sum_{\mu_1\ldots \mu_n} P_{\mu_1 \ldots \mu_n} \sigma_{\mu_1}\otimes \cdots\otimes \sigma_{\mu_n},$$

\nl  where $\mu_i \in \{0,x,y,x\}$ for each $i \in \{1,\ldots ,n\}$. The real expansion coefficients $P_{\mu_1 \ldots
\mu_n}$ are given by $P_{\mu_1 \ldots \mu_n} =\mathtt{tr}(\sigma_{\mu_1}\otimes \cdots \otimes \sigma_{\mu_n}\rho)$. 
Since the eigenvalues of the Pauli matrices are $\pm 1$, the expansion coefficients satisfy the inequality $\vert P_{\mu_1 \ldots \mu_n} \vert \leq 1 $. In what follows, for sake of simplicity, we will use without distinction $I$ or $\sigma_0$.
We denote by $\mathcal{ D}(\otimes^n \mathbb{C}^2)$ the set of all density operators of $\otimes^n\mathbb{C}^2$; hence the set $\mathcal{D} =\bigcup_{i\in N} \mathcal{D}(\otimes^n \mathbb{C}^2)$ will be the set of all possible density operators. Moreover, we can identify in each space $\mathcal{D}(\otimes^n \mathbb{C}^2)$ two special operators $P_0^{(n)} = \frac{1}{2^n} I^{n-1}\otimes P_0 $ and $P_1^{(n)} = \frac{1}{2^n} I^{n-1}\otimes P_1 $ that represent, in this framework, the falsity-property and the truth-property, respectively. 
The \emph{probability of truth} $\mathtt{p}$ of a density operator $\rho$ is dictated by the Born rule \cite{beltrametti2014quantum-p1,beltrametti2014quantum-p2} and equals 
$$\mathtt{p}(\rho) = \mathtt{tr}(P_1^{(n)} \rho).$$
\nl In case $\rho = \vert \psi \rangle \langle \psi \vert$, where $\vert \psi \rangle = c_0\vert 0 \rangle + c_1 \vert 1 \rangle$, then $\mathtt{p}(\rho) = \vert c_1 \vert ^2$. 
\nl Let $\rho \in \mathcal{D}( \mathbb{C}^2)$. Then $\rho$ can be represented as a linear superposition $\rho = \frac{1}{2}(I + r_x\sigma_x + r_y\sigma_y + r_z\sigma_z)$, where $r_x$, $r_y$, $r_z$ are real numbers such that $r_x^2 + r_y^2 + r_z^2 \leq 1$. Therefore, every density operator $\rho$ in $\mathcal{D}(\mathbb{C}^2)$ has the matrix representation: 
\begin{equation}\label{uno}
\rho = \frac{1}{2}
\left(\begin{array}{cc}
1+r_z & r_x -ir_y \\
r_x +ir_y & 1-r_z
\end{array}\right)
=
\left(\begin{array}{cc}
1- \alpha & \beta \\
\beta^*   & \alpha
\end{array}\right)
\end{equation}
\nl Furthermore, any real number $\lambda$ ($0\leq \lambda \leq 1$) uniquely determines a density operator as follows:
\begin{equation}\label{due}
\rho_{\lambda} = (1- \lambda)P_0 + \lambda P_1 = \frac{1}{2}(I + (1-2\lambda) \sigma_z) =
\left(\begin{array}{cc}
1- \lambda & 0 \\
0 & \lambda
\end{array}\right) 
\end{equation}
In virtue of (\ref{uno}) and (\ref{due}), one may verify that, whenever $\rho \in \mathcal{D}(\mathbb{C}^2)$, then $\mathtt{p}(\rho) = \frac{1-r_z}{2}$ and $\mathtt{p}(\rho_\lambda) = \lambda$. Thus each density operator $\rho$ in $\mathcal{D}(\mathbb{C}^2)$ can be written as

\begin{equation}
\rho =
\left(\begin{array}{cc}
1-\mathtt{p}(\rho) & a \\
a^* & \mathtt{p}(\rho)
\end{array}\right) 
\end{equation}

\nl In the usual model of quantum computation the state of a system is pure and the operations ({\it quantum gates}) are represented by unitary operators. Nevertheless, in case a system is not completely isolated from the
environment its evolution is, in general, irreversible. A model of quantum computing that relates to that phenomenon is mathematically described by means of {\it quantum operations} (as quantum gates) acting on density operators (as information quantities).

\nl Given a finite dimensional complex Hilbert space $H$, we will denote by $\mathcal{ L}(H)$ the vector space of all linear operators on $H$. Let $H_1,H_2$ be two finite dimensional complex Hilbert spaces. A {\it super operator } is a linear operator $\mathcal{ E}: \mathcal{ L}(H_1)\rightarrow \mathcal{ L}(H_2)$ sending density operators to density operators \cite{AB}. This is equivalent to say that $\mathcal{ E}$ is trace-preserving and positive, i.e. sends positive semi-definite Hermitian operators to positive semi-definite Hermitian operators. A super operator $\mathcal{ E}$ is said to be a {\it quantum operation} iff the super operator $\mathcal{ E}\otimes I_H $ is positive, where $I_H $ is the identity super operator on an arbitrary finite dimensional complex Hilbert space $H$. In this case $\mathcal{ E}$ is also called {\it completely positive}. The following theorem, dued to K. Kraus \cite{K},  provide an equivalent definition of quantum operations:

\begin{theo}\label{kraus}
A linear operator $\mathcal{ E}: \mathcal{ L}(H_1)\rightarrow \mathcal{ L}(H_2)$ is a quantum operation iff $\forall \rho \in \mathcal{ L}(H_1)$:

$$\mathcal{ E}(\rho) = \sum_i A_i \rho A_i^{\dagger}$$ for some set of operators $\{A_i\}$ such that $\sum_i A_i^{\dagger} A_i = I $.
\end{theo}

\section{A probabilistic Stone-Weierstrass type theorem}\label{sw}

\nl The aim of the present section is to propose a representation, in probabilistic terms, of a particular class of polynomials via quantum operations. Such a result will be expedient to prove a probabilistic Stone-Weierstrass type theorem. First of all, let us introduce some notations and preliminary definitions.
The term {\it multi-index} denotes an ordered $n$-tuple $\alpha = (\alpha_1, \ldots, \alpha_n)$ of non negative integers $\alpha_i$. The
{\it order} of $\alpha$  is given by $\vert \alpha \vert = \alpha_1 + \ldots \ldots + \alpha_n$. If ${\bf x} = (x_1, \ldots, x_n)$ is an $n$-tuple of
variables and $\alpha = (\alpha_1, \ldots, \alpha_n)$ a multi-index, the monomial ${\bf x}^{\alpha}$ is defined by ${\bf x}^{\alpha} = x_1^{\alpha_1}x_2^{\alpha_2} \ldots  x_n^{\alpha_n}$. In this language a real polynomial of order $k$ is a function $P({\bf x}) = \sum_{\mid \alpha \mid \leq k} a_\alpha {\bf x}^\alpha$ such that $a_\alpha \in \mathbb{R}$.

\nl Let ${\bf x} = (x_1, \ldots, x_n)$ and $k$ be a natural number. If we define the set $D_k({\bf x})$ as follows:

\begin{equation}
 D_k({\bf x}) = \{(1-x_1)^{\alpha_1} x_1^{\beta_1} \ldots (1-x_n)^{\alpha_n} x_n^{\beta_n} : \alpha_i + \beta_i = k, \hspace{0.2cm} i \in \{1,\ldots ,n\} \},
\end{equation}

\nl then we obtain the following useful lemmas:

\begin{lem}\label{KRO1}
Let ${\bf X}_1, \ldots ,{\bf X}_n$ be a family of matrices such that 
$$ {\bf X}_i = \left(\begin{array}{cc}
1-x_i & b_i \\
b_i^* & x_i
\end{array}\right)
$$ and let ${\bf X} = (\otimes^k{\bf X}_1)\otimes (\otimes^k{\bf X}_2)\otimes \cdots \otimes (\otimes^k{\bf X}_n) $. Then 
$$Diag({\bf X}) = D_{k}(x_1,\ldots,x_n).$$

\end{lem}

\begin{proof}
It can be verified that $\otimes^{k}X_{i}= \{h_1 h_2 \ldots  h_k: h_j \in \{(1-x_i),x_i\}, 1\leq j \leq k\}= \{(1-x_1)^{\alpha} x_1^{\beta} : \alpha + \beta = k \}$. Thus, $(\otimes^k{\bf X}_1)\otimes (\otimes^k{\bf X}_2)\otimes \cdots \otimes (\otimes^k{\bf X}_n) =\{(1-x_1)^{\alpha_1} x_1^{\beta_1} \ldots (1-x_n)^{\alpha_n} x_n^{\beta_n} : \alpha_i + \beta_i = k, \hspace{0.2cm} i \in \{1,\ldots ,n\} \}$. Whence our claim follows.
\end{proof}

\begin{lem}\label{KRO2}
Let ${\bf x} = (x_1, \ldots, x_n)$ and $k$ be a natural number. For any monomial ${\bf x}^{\alpha}$, such that $\mid \alpha \mid \leq k$, the following conditions hold: 
\begin{enumerate}
\item
${\bf x}^\alpha = \sum_{{\bf y} \in D_k({\bf x})} \delta_{{\bf y}}{\bf y}$;

\item 
$1-{\bf x}^{\alpha} = \sum_{{\bf y} \in D_k({\bf x})} \gamma_y{\bf y}$;    
\end{enumerate}

where $\delta_y$ and $\gamma_y$ are in $\{0,1\}$.
\end{lem}

\begin{proof}
First, let us define, for any $i \in \{1, \ldots, n\}$, a matrix ${\bf X_i}$ as follows 
$$ {\bf X}_i = \left(\begin{array}{cc}
1-x_i & 0 \\
0 & x_i
\end{array}\right)$$

\nl 1) Let ${\bf x}^{\alpha} = x_1^{\alpha_1}x_2^{\alpha_2} \ldots x_n^{\alpha_n}$ such that $\mid \alpha \mid \leq k$. Thus, there are $s_1, \ldots, s_n$ such that $\alpha_i + s_i = k$. Let ${\bf W} = (\otimes^{s_1}{\bf X}_1)\otimes (\otimes^{s_2}{\bf X}_2)\otimes \cdots \otimes (\otimes^{s_n}{\bf X}_n)$ and consider the matrix ${\bf W}{\bf x}^{\alpha}$. In view of Lemma \ref{KRO1}, $Diag({\bf W}{\bf x}^{\alpha}) \subseteq D_{k}(x_1,\ldots,x_n)$ since every element in $Diag({\bf W}{\bf x}^{\alpha})$ is a monomial of order $nk$. Further, since $\mathtt{tr}({\bf W}{\bf x}^{\alpha})= (\mathtt{tr}{\bf W}){\bf x}^{\alpha}= 1{\bf x}^{\alpha}={\bf x}^{\alpha}$, we obtain ${\bf x}^\alpha = \mathtt{tr}({\bf W}{\bf x}^{\alpha})$, i.e. the required polynomial expansion.

\nl2) Let ${\bf X} = (\otimes^k{\bf X}_1)\otimes (\otimes^k{\bf X}_2)\otimes \cdots \otimes (\otimes^k{\bf X}_n) $. By Lemma \ref{KRO1}
$Diag({\bf X}) = D_{k}(x_1,\ldots,x_n)$ and $\mathtt{tr}({\bf X}) = 1$. Upon recalling that ${\bf x}^\alpha = \sum_{{\bf y} \in D_k({\bf x})} \delta_{{\bf y}}{\bf y} $, we define $\gamma_{{\bf y}} = 1$ if $\delta_{{\bf y}} = 0$ and $\gamma_{{\bf y}} = 0$ if $\delta_{{\bf y}} = 1$. Therefore, 

\begin{eqnarray*}
1& =& \mathtt{tr}({\bf X}) \\
& =& \sum_{{\bf y} \in D_k({\bf x})} \delta_{{\bf y}}{\bf y} + \sum_{{\bf y} \in D_k({\bf x})} \gamma_y{\bf y} \\
& =& {\bf x}^\alpha + \sum_{{\bf y} \in D_k({\bf x})} \gamma_y{\bf y}
\end{eqnarray*}
and  $1 - {\bf x}^{\alpha} = \sum_{{\bf y} \in D_k({\bf x})} \gamma_y{\bf y}$

\end{proof}

\nl In virtue of the previous claims, we can prove a technical but rather important theorem:

\begin{theo}\label{KRO3}
Let  ${\bf x} = (x_1, \ldots ,x_n)$ be an $n$-tuple of variables, and let $P({\bf x})= \sum_{{\bf
y} \in D_k({\bf x})} a_{{\bf y}} {\bf y}$ be a polynomial such that ${\bf y} \in D_k({\bf x})$, $0 \leq a_{{\bf y}} \leq 1$ and the restriction $P({\bf x}) \upharpoonright_{[0,1]^n}$ be such that $0 \leq P({\bf x}) \upharpoonright_{[0,1]^n} \leq 1$.
\nl Then there exists a polynomial quantum operation $\mathcal{ P}: \mathcal{ L} (\otimes^{nk} \mathbb{C}^2)
\rightarrow \mathcal{ L} (\otimes^{nk} \mathbb{C}^2)$ such that, for
any n-tuple $\sigma = (\sigma_1, \ldots, \sigma_n)$ in $\mathcal{
D}(\mathbb{C}^2)$, $$\mathtt{p}(\mathcal{ P}((\otimes^k \sigma_1)
\otimes \cdots \otimes  (\otimes^k \sigma_n))) =
P(\mathtt{p}(\sigma_1), \ldots, \mathtt{p}(\sigma_n)).$$
 Moreover,
$$\mathcal{ P}((\otimes^k \sigma_1) \otimes \cdots \otimes  (\otimes^k
\sigma_n)) = (\frac{1}{2^{nk-1}} \otimes^{nk-1} I) \otimes
\rho_{P(\mathtt{p}(\sigma_1), \ldots, \mathtt{p}(\sigma_n))} .$$\\

\end{theo}

\begin{proof}
Let ${\sigma}_1, \ldots, {\sigma}_n$ be density operators on $\mathbb{C}^2$. Assume that for any ${\sigma}_i$
$$ {\sigma}_i = \left(\begin{array}{cc}
1-x_i & b_i \\
b_i^* & x_i
\end{array}\right)\\
$$
\nl Hence, $\mathtt{p}(\sigma_i) = x_i$. Evidently, $\sigma = (\otimes^k \sigma_1)\otimes  \cdots \otimes (\otimes^k \sigma_n)$ is a  $2^{nk} \times 2^{nk}$ matrix and, by Lemma \ref{KRO1}, $Diag(\sigma) = D_k(x_1,\ldots,x_n)$.  Thus,  each ${\bf y} \in D_k({\bf x})$ can be seen as the $(i,i)$-th entry of $Diag(\sigma)$. Further, the polynomial $P({\bf x}) = \sum_{{\bf y} \in D_k({\bf x})} a_{{\bf y}} {\bf y} = \sum_{j=1}^{2^{nk}}a_j {\bf y}_j $ is such that every ${\bf y}_j$ is the $(j,j)$-th entry of $Diag(\sigma)$.  Let, now, ${\bf y}_{j_0} \in Diag(\sigma)$.

\bigskip

a) We want to place the elements of the form $a_{j_0}{\bf y}_{j_0}$ in the $(2s, 2s)$-th
entries of a $2^{nk} \times 2^{nk}$ matrix.
\\ Let us consider the
$2^{nk} \times 2^{nk}$ matrix $A_{j_0}^{2s} =
\sqrt{\frac{a_{j_0}}{2^{nk-1}}} D_{j_0}^{2s}$  such that
$D_{j_0}^{2s}$ has $1$ just in the $(2s, j_0)$-th entry and $0$
in any other entry. One may verify that $A_{j_0}^{2s} \sigma (A_{j_0}^{2s})^\dagger$ is the required
matrix. Moreover:
$$
\sum_{2s}A_{j_0}^{2s} \sigma (A_{j_0}^{2s})^\dagger = \frac{1}{2^{nk-1}}
\left(\begin{array}{ccccc}
0 & 0 & 0 & 0 & \ldots \\
0 & a_{j_0}{\bf y}_{j_0} & 0 & 0 & \ldots \\
0 & 0 & 0 & 0 & \ldots \\
0 & 0 & 0 & a_{j_0}{\bf y}_{j_0} & \ldots  \\
\vdots & \vdots & \vdots & \vdots & \ddots \\
\end{array}\right)
$$
\bigskip

b) Let us recall that $ 1 =  \sum_{{\bf y} \in D_k({\bf x})}
{\bf y} = \sum_{j=1}^{2^{nk}}{\bf y}_j $. Then:

$$1 -\sum_{j=1}^{2^{nk}}a_j {\bf y}_j = \sum_{j=1}^{2^{nk}}{\bf y}_j -
\sum_{j=1}^{2^{nk}}a_j {\bf y}_j = \sum_{j=1}^{2^{nk}} (1-
a_j){\bf y}_j.$$

\nl We now want to stick the elements of the form $(1 - a_{j_0}){\bf y}_{j_0}$ into the $(2s-1,
2s-1)$-th entries of a $2^{nk} \times 2^{nk}$ matrix. Let us
consider the $2^{nk} \times 2^{nk}$ matrix $A_{j_0}^{2s-1} =
\sqrt{\frac{1- a_{j_0}}{2^{nk-1}}} D_{j_0}^{2s-1}$  such that
$D_{j_0}^{2s-1}$ have $1$ just in the $(2s-1, j_0)$-th entry and
$0$ in any other entry.
Again, one may verify that 
$A_{j_0}^{2s-1} \sigma (A_{j_0}^{2s-1})^\dagger$ is the required
matrix.
 Furthermore:
 
\begin{eqnarray*}
\sum_{2s-1}A_{j_0}^{2s-1} \sigma (A_{j_0}^{2s-1})^\dagger =
\frac{1}{2^{nk-1}}
\left(\begin{array}{ccccc}
(1 - a_{j_0}){\bf y}_{j_0} & 0 & 0 & 0 & \ldots \\
0 & 0 & 0 & 0 & \ldots \\
0 & 0 & (1 - a_{j_0}){\bf y}_{j_0} & 0 & \ldots \\
0 & 0 & 0 & 0 & \ldots  \\
\vdots & \vdots & \vdots & \vdots & \ddots \\
\end{array}\right)
\end{eqnarray*}

\vspace{0.5cm}

\noindent Thus, some calculations show that:

\begin{eqnarray*}
 \mathcal{ P} &=&\sum_{j_0}\sum_{2s}A_{j_0}^{2s} \sigma (A_{j_0}^{2s})^\dagger +
\sum_{j_0}\sum_{2s-1}A_{j_0}^{2s-1} \sigma
(A_{j_0}^{2s-1})^\dagger \\
&=&(\frac{1}{2^{kn-1}} \otimes^{nk-1} I) \otimes \left(\begin{array}{cc}
1-\sum_{j=1}^{2^{nk}}a_j {\bf y}_j & 0 \\
0 & \sum_{j=1}^{2^{nk}}a_j {\bf y}_j
\end{array}\right)
\end{eqnarray*}

%
%
%
%
%
\vspace{0.4cm}

\nl Now, set $A = \sum_{j_0}\sum_{2s}(A_{j_0}^{2s})^\dagger
A_{j_0}^{2s} + \sum_{j_0}\sum_{2s+1}(A_{j_0}^{2s+1})^\dagger
A_{j_0}^{2s+1}$. Our task is to verify that $A = I$.

\bigskip

c) First of all, notice that the matrix $(A_{j_0}^{2s})^\dagger
A_{j_0}^{2s}$ has the value $\frac{a_{j_0}}{2^{nk-1}}$ just in the
$(j_0, j_0)$-th entry, while any other entry is $0$. Therefore, the
matrix $\sum_{2s}(A_{j_0}^{2s})^\dagger A_{j_0}^{2s}$ has the
value $\frac{2^{nk-1} a_{j_0}}{2^{nk-1}} = a_{j_0}$ in the $(j_0,
j_0)$-th entry and all the other entries equal $0$. Hence:
$$ \sum_{j_0}\sum_{2s}(A_{j_0}^{2s})^\dagger A_{j_0}^{2s} =
\left(\begin{array}{ccccc}
a_1 & 0 & 0 & 0 & \ldots \\
0 & a_2  & 0 & 0 & \ldots \\
0 & 0 & a_3 & 0 & \ldots \\
0 & 0 & 0 & a_4 & \ldots  \\
\vdots & \vdots & \vdots & \vdots & \ddots \\
\end{array}\right)
$$
\bigskip

d) On the other hand the matrix $(A_{j_0}^{2s-1})^\dagger
A_{j_0}^{2s-1}$ has the value $\frac{1- a_{j_0}}{2^{nk-1}}$ just
in the $(j_0, j_0)$-th entry and $0$ in any other. Therefore, the matrix $\sum_{2s-1}(A_{j_0}^{2s-1})^\dagger
A_{j_0}^{2s-1}$ has the value
$\frac{2^{nk-1}(1-a_{j_0})}{2^{nk-1}} = 1- a_{j_0}$ in the $(j_0,
j_0)$-th entry and any other is $0$. Hence:
$$ \sum_{j_0}\sum_{2s-1}(A_{j_0}^{2s-1})^\dagger A_{j_0}^{2s-1} =\left(\begin{array}{ccccc}
1-a_1 & 0 & 0 & 0 & \ldots \\
0 & 1- a_2  & 0 & 0 & \ldots \\
0 & 0 & 1- a_3 & 0 & \ldots \\
0 & 0 & 0 & 1- a_4 & \ldots  \\
\vdots & \vdots & \vdots & \vdots & \ddots \\
\end{array}\right)
$$
\bigskip

\nl Thus, $\sum_{j_0}\sum_{2s}(A_{j_0}^{2s})^\dagger A_{j_0}^{2s-1} +
\sum_{j_0}\sum_{2s-1}(A_{j_0}^{2s-1})^\dagger A_{j_0}^{2s-1} = I$. Whence,  by Theorem \ref{kraus}, $\mathcal{ P}$ is a quantum operation.
\end{proof}

\vspace{0.3cm}

\nl In virtue of Theorem \ref{KRO3}, we can now establish a probabilistic version of the classical Stone-Weierstrass theorem.

\begin{theo}\label{KRO4}
Let ${\bf x} = (x_1, \ldots, x_n)$ be an $n$-tuple of variables and $f:[0,1]^n \rightarrow (0,1)$ be a continuous function. Then, for each $\epsilon > 0 $ there exists a quantum operation  $\mathcal{ P}: \mathcal{ L} (\otimes^{nk} \mathbb{C}^2) \rightarrow \mathcal{ L} (\otimes^{nk} \mathbb{C}^2)$ and a constant $M \geq 1$ such that, for any density matrix $\sigma =(\otimes^k \sigma_1) \otimes \cdots \otimes  (\otimes^k \sigma_n)$, the following inequality holds: 
$$\mid \mathtt{p}(\mathcal{ P}(\sigma)) - \frac{1}{M} f(\mathtt{p}(\sigma_1), \ldots, \mathtt{p}(\sigma_n) ) \mid \leq \epsilon.$$
\end{theo}

\begin{proof}
Let $f:[0,1]^n \rightarrow (0,1)$ be a continuous function. By the classical Stone-Weierstrass theorem, there exists a polynomial $P_0(x_1, \ldots, x_n) = \sum_{\mid \alpha \mid \leq k }a_{\alpha}{\bf x}^\alpha $ such that for each $ \epsilon > 0 $, $\mid P_0-f  \mid \leq \frac{\epsilon}{2} $. Let $a_{\alpha_1}, \ldots, a_{\alpha_n} $ be positive coefficients and $a_{\beta_1}, \ldots, a_{\beta_s} $ be negative coefficients in the polynomial $P_0(x_1, \ldots, x_n)$. Let $M$ be a positive real number such that $\frac{\sum_{i=1}^{s}\mid a_{\beta_i} \mid}{M} \leq \frac{\epsilon}{2}$. Let us define a polynomial $P$ by $P(x_1, \ldots, x_n) = \sum_{i=1}^n \frac{a_{\alpha_i}}{M}{\bf x}^{\alpha_i} + \sum_{j=1}^s \frac{\mid a_{\beta_j}\mid}{M}(1- {\bf x}^{\beta_j}) $. Then, we obtain that $P(x_1, \ldots, x_n) = \frac{1}{M} P_0(x_1, \ldots, x_n) +  \frac{\sum_{i=1}^{s}\mid a_{\beta_i} \mid}{M}$. Therefore, in $[0,1]^n$:

\begin{eqnarray*}
 \mid P - \frac{1}{M}f  \mid  &= &\mid P - \frac{1}{M}P_0 +\frac{1}{M} P_0 - \frac{1}{M}f  \mid\\
& \leq &\mid P - \frac{1}{M}P_0 \mid + \mid \frac{1}{M} P_0 - \frac{1}{M}f  \mid\\
&  \leq & \frac{\epsilon}{2} + \frac{\epsilon}{M}\\& \leq& \epsilon 
\end{eqnarray*}

\nl So, by Lemma \ref{KRO2}, we obtain that $P({\bf x}) = \sum_{{\bf y} \in D_k({\bf x})} a_{{\bf y}} {\bf y}$ is such that  $a_{{\bf y}} \geq 0$ and $0\leq P({\bf x}) \leq 1$ in $[0,1]^n$. Whence, by Theorem \ref{KRO3}, there exists a quantum operation $\mathcal{ P}:\mathcal{ L}(\otimes^{nk} \mathbb{C}^2) \rightarrow \mathcal{ L}(\otimes^{nk}\mathbb{C}^2)$ associated to $P$ such that for each n-tuple $\sigma = (\sigma_1, \ldots, \sigma_n)$, with $\sigma_i$ in $\mathcal{D}(\mathbb{C}^2)$, $\mathtt{p}(\mathcal{ P}((\otimes^k \sigma_1) \otimes \cdots \otimes  (\otimes^k \sigma_n))) = P(x_1 /\mathtt{p}(\sigma_1), \ldots, x_n/\mathtt{p}(\sigma_n))$.\footnote{By $x_i/\mathtt{p}(\sigma_i)$ we mean the attribution of the value $\mathtt{p}(\sigma_i)$ to the variable $x_{i}$. } Thus $\mid \mathtt{p}(\mathcal{ P}((\otimes^k \sigma_1) \otimes \cdots \otimes  (\otimes^k \sigma_n)))) - \frac{1}{M} f(\mathtt{p}(\sigma_1), \ldots, \mathtt{p}(\sigma_n) ) \mid \leq \epsilon$.
\end{proof}

%

\section{Representing the standard PMV-operations}\label{rcn}
\nl In this section we apply the results obtained to two functions (namely, the product t-norm $\bullet$, and the \L ukasiewicz conorm $\oplus$) of definite importance in fuzzy logic. Let us recall some notions first.
\nl The \emph{standard PMV-algebra} \cite{EGM, MR} is the algebra $$[0,1]_{PMV} = \langle [0,1], \oplus, \bullet, \neg, 0, 1\rangle,$$ where $[0,1]$ is the real unit segment, $x \oplus y = \min(1, x+y)$, $\bullet$ is the real product, and $\neg x = 1-x$. This structure plays a notable role in quantum computing, in that it decribes, in a probabilistic way, a relevant system of quantum gates named {\it Poincar\`{e} irreversible quantum computational algebra} \cite{CDGL1, DF}.   
 The connection between quantum computational logic with mixed states and fuzzy logic, comes from the election of a system of quantum gates such that, when interpreted under probabilistic semantics, they turn out in some kind of operation in the real interval $[0,1]$. The above-mentionated $PMV-algebra$ is a structure that represents algebraic counterpart of the probabilistic semantics conceived from the continuous $t$-norm. On the other hand, the use of fuzzy logics (and infinite-valued \L ukasiewicz logic in particular) in game theory and theoretical physics was pioneered in \cite{M1,M2}, linking the mentioned structures with Ulam games and $AF-C^*$-algebras, respectively. We will pay special attention to the study on the \L ukasiewicz $t$-norm, due to its relation with Ulam games and its possible applications to error-correction codes in the context of quantum computation.
\nl Evidently, $\neg$ can be expressed as a polynomial in the generator system $D_1(x)$; whence by Theorem \ref{KRO3}, it is representable as a polynomial quantum operation. A possible representation can be the following: $NOT(\rho) = \sigma_x \rho \sigma^{\dagger}_x$. In fact, $\mathtt{p}(NOT(\rho)) = 1 - \mathtt{p}(\rho)$.

\nl Furthermore, $\bullet$ can be represented by a polynomial in the generator system $D_2(x,y)$. According with the construction presented in Theorem \ref{KRO3}, the following representation obtains. Let us consider the following matrices:
 {\Small
\begin{eqnarray*}
 G_1 =
\left(\begin{array}{cccc}
\frac{1}{\sqrt{2}} & 0 & 0 & 0 \\
0 & 0 & 0 & 0 \\
0 & 0 & 0 & 0 \\
0 & 0 & 0 & 0 \\
\end{array}\right)
& G_2 =
\left(\begin{array}{cccc}
0 & \frac{1}{\sqrt{2}} & 0 & 0 \\
0 & 0 & 0 & 0 \\
0 & 0 & 0 & 0 \\
0 & 0 & 0 & 0 \\
\end{array}\right)
&G_3 =
\left(\begin{array}{cccc}
0 & 0 & \frac{1}{\sqrt{2}} & 0 \\
0 & 0 & 0 & 0 \\
0 & 0 & 0 & 0 \\
0 & 0 & 0 & 0 \\
\end{array}\right)\\
G_4 =
\left(\begin{array}{cccc}
0 & 0 & 0 & 0 \\
0 & 0 & 0 & 0 \\
\frac{1}{\sqrt{2}} & 0 & 0 & 0 \\
0 & 0 & 0 & 0 \\
\end{array}\right)
&G_5 =
\left(\begin{array}{cccc}
0 & 0 & 0 & 0 \\
0 & 0 & 0 & 0 \\
0 & \frac{1}{\sqrt{2}} & 0 & 0 \\
0 & 0 & 0 & 0 \\
\end{array}\right)
 &G_6 =
\left(\begin{array}{cccc}
0 & 0 & 0 & 0 \\
0 & 0 & 0 & 0 \\
0 & 0 & \frac{1}{\sqrt{2}} & 0 \\
0 & 0 & 0 & 0 \\
\end{array}\right)\\
G_7 =
\left(\begin{array}{cccc}
0 & 0 & 0 & 0 \\
0 & 0 & 0 & 0 \\
0 & 0 & 0 & 0 \\
0 & 0 & 0 & \frac{1}{\sqrt{2}} \\
\end{array}\right)
&G_8 =
\left(\begin{array}{cccc}
0 & 0 & 0 & 0 \\
0 & 0 & 0 & \frac{1}{\sqrt{2}} \\
0 & 0 & 0 & 0 \\
0 & 0 & 0 & 0 \\
\end{array}\right)&
\end{eqnarray*}
}
%

\vspace{0.2cm}

\noindent One may verify that $\sum_{i=1}^{8} G_i(\tau \otimes \sigma)G_i^{\dagger} =
\frac{1}{2}I \otimes \rho_{\mathtt{p}(\tau)\mathtt{p}(\sigma)}$
where $\sigma, \tau \in \mathcal{D}(\mathbb{C}^{2}) $.  Thus, by Kraus representation Theorem \cite{K}, $\sum_{i=1}^{8} G_i(\tau \otimes \sigma)G_i^{\dagger}$ is a quantum operation, and ${\mathtt p}(\sum_{i=1}^{8} G_i(\tau \otimes \sigma)G_i^{\dagger}) = {\mathtt p}(\tau) \bullet {\mathtt p}(\sigma)$. That quantum operation represents the well known quantum gate $IAND$ modulo a tensor power \cite{CDGL1, NIC}.
As regards the \L ukasiewicz conorm $\oplus$, it can be seen that it is not a polynomial, Figure \ref{fig:1}.

\begin{figure}[htbp]
\includegraphics[scale=0.5]{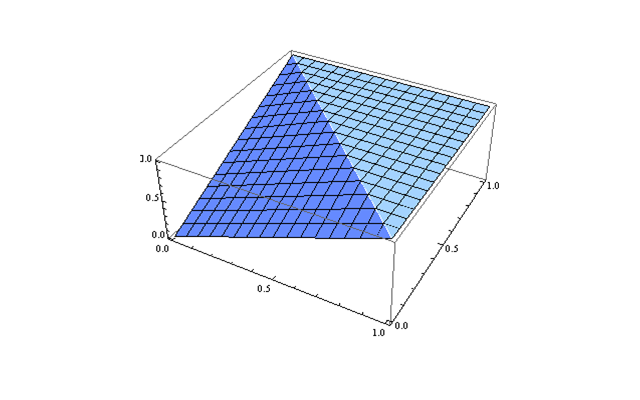}
\caption{The \L ukasiewicz conorm}\label{fig:1}
\end{figure}

\nl Therefore, our idea is to obtain a polynomial $P(x,y)$ in some generator system $D_k(x,y)$, such that $P(x,y)$ can approximate the \L ukasiewicz sum.\nl By using numerical methods, we get the following approximating polynomial of $\oplus$ in $[0,1]$: $$P(x,y) = \frac{5}{12}x(1-x)+\frac{5}{12}y(1-x)+\frac{5}{12}x(1-y)+\frac{5}{12}y(1-y)+\frac{1}{2}x+\frac{1}{2}y,$$ whose graph is depicted in Figure \ref{fig:2}.

\begin{figure}[htbp]
\includegraphics[scale=0.5]{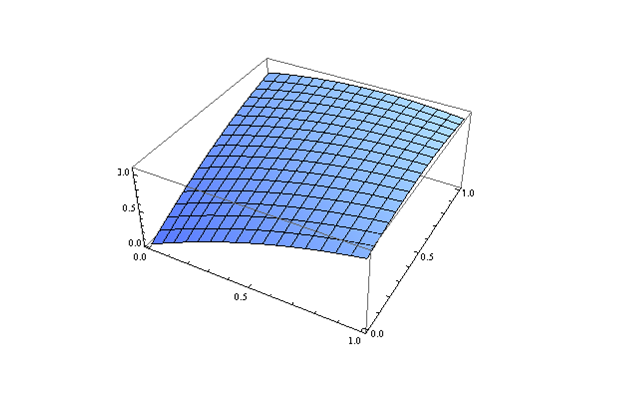}
\caption{$P(x,y)$}\label{fig:2}
\end{figure}

\begin{figure}[htbp]
\includegraphics[scale=0.5]{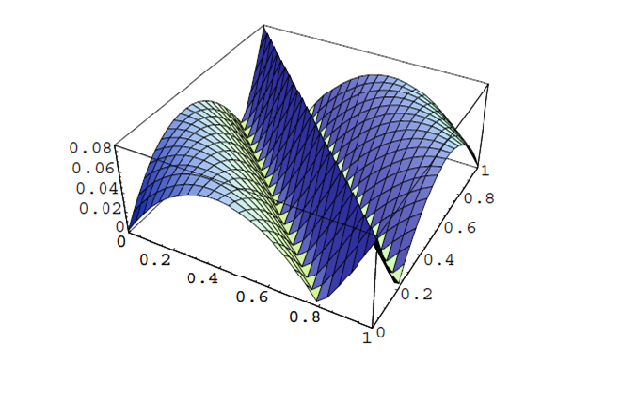}
\caption{$(x\oplus y) - P(x,y)$}\label{fig:3}
\end{figure}

\nl Let us remark that $0\leq P(x,y) \leq x\oplus y$. Therefore, $e = \max_{[0,1]}\{(x\oplus y) - P(x,y)\}\leq 0.08$, as Figure \ref{fig:3} shows. Furthermore, one readily realizes that $P(x,y)$ is a polynomial obtained from the generator system $D_2(x,y)$, that also satisfies the hypothesis of Theorem \ref{KRO3}. Thus, $P(x,y)$ is representable as a polynomial quantum operation $\mathcal{ P}_{\oplus}$, where $${\mathtt p}(\mathcal{ P}_{\oplus}(\tau \otimes \sigma)) = ({\mathtt p}(\tau) \oplus {\mathtt p}(\sigma)) \pm 0.08$$

\section{Conclusions and Open problems}\label{final}

\nl In virtue of the results in Section \ref{rcn}, it turns out that the approximation obtained in the case of PMV-algebras is definitely accurate. Further, in \cite{FSA}, authors show a covergence theorem that allows to achieve every degree of accuracy; the price to pay is the increasing of the degree of the approximating polynomial. In our opinion, that is an interesting achievement, since it provides a (quantum) computational motivation for the investigation of algebraic structures equipped with the \L ukasiewicz sum and, to a certain extent, it relates ``classical'' fuzzy logic to quantum computational logics.\\
\noindent Nonetheless, some general remarks on the whole construction are in order as well.

\begin{enumerate}
\item If one wants to apply our results to the models of (quantum) computing, efficiency is of central importance. Unfortunately, since the number of copies required in our construction corresponds to the degree of the approximating polynomial, it is impossible to generally specify the dimension of the Hilbert space required to achieve, given a certain $\epsilon$, the approximating polynomial.

\item Since the work of Ekert and other scholars \cite{KW, horo1, horo2, horo3}, a direct study of estimations of linear and non-linear functionals of (quantum) states using quantum networks has been proposed. This approach has the advantage that it bypasses quantum tomography, providing more direct estimations of both linear and non-linear functionals of a state. It could be of interest, in our opinion, to investigate if, and in what cases, our construction can be carried out by a quantum network.
\end{enumerate}




\end{document}